\documentclass{iopart}
\usepackage{iopams}
\usepackage{geometry}                
\usepackage{tikz}
\usetikzlibrary{calc}

\usepackage{graphicx}
\usepackage[english]{babel}    
\usepackage[latin1]{inputenc}      
\usepackage{epstopdf}
\usepackage{amsthm}
\usepackage{color}
\usepackage[pagebackref,pdfusetitle]{hyperref}
\usepackage{url}

\def\eqalign#1{\null\,\vcenter{\openup\jot \mathsurround=0pt \ialign{\strut
     \hfil$\displaystyle{##}$&$ \displaystyle{{}##}$\hfil \crcr#1\crcr}}\,}

\newtheorem{proposition}{Proposition}
\newtheorem{corollary}[proposition]{Corollary}
\newtheorem{definition}[proposition]{Definition}

\def\tempcolor{\color{black}}   
\def\EBcolor{\color{black}}

\def\eqref#1{(\ref{#1})}
\def\tfrac#1#2{{\textstyle \frac{#1}{#2}}}
\def\R{\mathbb{R}}

\begin{document}

\title{Integrability properties of  Kahan's method }

\author{Elena Celledoni$^1$, Robert I McLachlan$^2$, David I McLaren$^3$, Brynjulf Owren$^1$ and G R W Quispel$^3$}

\address{$^1$ 	Department of Mathematical Sciences,
	NTNU,
	7491 Trondheim,
	Norway\eads{\mailto{elenac@math.ntnu.no}, \mailto{bryn@math.ntnu.no}}}
\address{$^2$ 	Institute of Fundamental Sciences,
	Massey University,
	Private Bag 11 222, Palmerston North 4442, New Zealand\ead{\mailto{r.mclachlan@massey.ac.nz}}}
\address{$^3$ 	Department of Mathematics,
	La Trobe University,
	Bundoora, VIC 3083, Australia\eads{\mailto{d.mclaren@latrobe.edu.au}, \mailto{r.quispel@latrobe.edu.au}}}

\begin{abstract}
\noindent
We present several novel examples of integrable quadratic vector fields
for which Kahan's {\tempcolor discretization} method preserves integrability. Our examples include
generalized Suslov and Ishii systems, Nambu systems, Riccati systems, and the first Painlev\'e
equation. We also discuss how Manin transformations arise in Kahan
discretizations of certain vector fields.
\end{abstract}
 
\section{Introduction}

This paper is about quadratic ordinary differential equations (ODEs) and their discretization. ODEs can be either integrable or non-integrable. In this paper, we concentrate on the case of integrable {\it quadratic} ODEs, and their discretization by Kahan's method (sometimes also called the Hirota--Kimura method); see \cite{CMOQgeometricKahan, hirota00,kahan,ka-li, kimura00dot, petrera12oio,pfadler} and references therein.
In \cite{CMOQgeometricKahan} it was shown that the {\em Kahan method} applied to the differential equation
\begin{equation*}\dot{x}=f(x), \quad x(0)=x_0,\end{equation*}
with quadratic vector field $f$ on $\R^n$,
coincides with the Runge--Kutta method
\begin{equation}
\label{eq:kahanrk}
\frac{x'-x}{h} = -\frac{1}{2}f(x) + 2 f\Big(\frac{x+x'}{2}\Big) - \frac{1}{2}f(x').
\end{equation}
As already noted by Kahan \cite{ka-li}, the Kahan method also coincides
with a certain so-called Rosenbrock method on quadratic vector fields, {\tempcolor which} reads as
\begin{equation}
\label{eq:rosenbrock}
\frac{x' - x}{h} = \Big(I - \frac{h}{2}f'(x)\Big)^{-1}f(x),
\end{equation}
or {\tempcolor equivalently}
\begin{equation}
\label{eq:rosenbrock2}
\frac{x' - x}{h} = \Big(I + \frac{h}{2}f'(x')\Big)^{-1}f(x').
\end{equation}
In \cite{CMOQgeometricKahan} it was also shown that when applied to systems {\tempcolor with a cubic Hamiltonian}
the Kahan method has a conserved quantity given by the modified Hamiltonian
\begin{equation}
\label{eq:Ht}
\widetilde H(x) := H(x) + \frac{1}{3}h \nabla H(x)^T \Big( I - \frac{1}{2}h f'(x)\Big)^{-1}f(x),
\end{equation}
and it preserves the measure 
\begin{equation*}
 \frac{dx_1\wedge\dots\wedge dx_n}{\det(I - \frac{1}{2} h f'(x))}.
\end{equation*}
This implies that when applied to quadratic Hamiltonian vector fields in $\mathbb{R}^2$ the Kahan method yields an integrable map.

In this paper, we will be {\tempcolor particularly} interested in the question whether the Kahan discretization of certain classes of integrable ODEs preserves integrability, i.e., whether the discrete systems obtained from these integrable ODEs are themselves again integrable. Because the Kahan map is rational, the integrability criterion we use in
this paper is that of degree growth (or, equivalently, that of algebraic entropy).
This criterion states that if the maximal degrees of the numerator and
denominator of the iterates of a rational map exhibit polynomial growth
(rather than the generic exponential growth), the map is conjectured to be
integrable \cite{bellon1999algebraic,veselov1992growth}. The salient feature is that the degrees of numerator and
denominator here are calculated after cancellation of common factors, thus
implying that integrable maps undergo a magically large number of such
cancellations.
 (Some publications on integrable ordinary {\it difference} equations are \cite{QRT1, QRT2, QCPN,suris03tpo}).
In the current paper, we {\tempcolor give some new examples and significantly generalize some of} the examples in \cite{petrera12oio}. 
Our method is to first identify a candidate system or class of systems for which the Kahan map satisfies the algebraic entropy criterion, and then to attempt to prove the systems' integrability by finding a conserved measure and a sufficient number of first integrals.

A second feature of our study is that, in view of the fact that the Kahan method is affine-covariant \cite{mclachlan1998numerical}, we attempt to express all results in a linear- or affine-covariant way and to work with linear- or affine-invariant classes of systems where possible.


Our outline is as follows. Section~\ref{section2} is on the application of Kahan's method to {\tempcolor \it{autonomous}} quadratic ODEs. Subsection~\ref{subsectionManin}  
treats the discretization of quadratic Hamiltonian ODEs in $\mathbb{R}^2$, and its connection to Manin transformations \cite{duistermaat10dis}. Subsection~\ref{subs:nambu} discusses the discretization of quadratic ODEs in $\mathbb{R}^3$ possessing two quadratic integrals. Subsection~\ref{subs:suslov} is about the discretization of generalized Suslov sustems in $\mathbb{R}^n$, including the integrable cases $n=2$ and $n=3$ \cite{petrera12oio,suslov46tm}. {\tempcolor Section 3 is on the application of Kahan's method to {\it nonautonomous} quadratic ODEs. Subsection 3.1 treats the discretization of certain Riccati equations.}
Last but not least, Section~\ref{subs:painleve} covers the (non-autonomous) first Painlev\'e equation \cite{gromak08pde}, whose Kahan discretization is also shown to be integrable. In this striking example the integrals are not rational functions and the mapping is not integrable in terms of elliptic or hyperelliptic functions.

\section{Autonomous problems}
\label{section2}
\subsection{Manin transformations}
\label{sec:2.1}
\label{subsectionManin}
Let $H(x,y)$ be an arbitrary cubic Hamiltonian. In \cite{CMOQgeometricKahan} it was shown that Kahan's method applied to the ODE
\begin{equation}
\label{eq:1}
\frac{d}{dt}\left(\matrix{x\cr y}\right)=
\left(
\matrix{
0 & 1\cr
-1 & 0}\right) \nabla H(x,y)
\end{equation}
yields an integrable map of $\mathbb{R}^2$. The question thus arises whether this integrable map is new or is already known. Here we show for one asymmetric Hamiltonian as well as for all symmetric Hamiltonians ($H(x,y)=H(y,x)$) that the integrable map obtained  by applying Kahan's method to (\ref{eq:1}) is in fact a so-called Manin transformation as in \cite[Chapter 4.2]{duistermaat10dis}, a known class of integrable maps of the plane. We briefly recall its definition.
\begin{definition} \label{def1}
Consider an elliptic curve $\gamma$ in the projective plane with a fixed point $O\in \gamma$. The  involution
$I_O: \gamma\rightarrow \gamma$ is defined as follows: For $P\neq O$, $I_O(P)$  is the unique third point of intersection of the line through $O$ and $P$ with $\gamma$, counting with multiplicities.  
For $P=O$, $I_O(P)$ is the third intersection point of the line tangent to $\gamma$ at $O$ with $\gamma$, again counting with multiplicities
 \footnote{ Such a point is guaranteed to exist {\EBcolor and be unique} because the curve is elliptic.}.
\end{definition}
We next consider a pencil of elliptic curves, generally denoted $P_z(x,y)=z_0 P^{0}(x,y)+z_1 P^1(x,y) = 0$. A base point $B$
is a point for which both $P^0$ and $P^1$ vanish, and must clearly be common to all curves in the pencil since
$P_z(B)=0$ for all $z$. Suppose  now that $B$ is a base point and that $P$ is not a base point. Then the curve to which $P$ belongs is uniquely given. We can therefore extend the definition of $I_B$ above to the whole pencil, base points excluded. Also, compositions of such involutions are well defined in this sense.


\begin{definition} Let $B_1$ and $B_2$ be base points in the projective plane.
A
Manin transformation is the composition $I_{B_1} \circ I_{B_2}$, defined on some subset of a pencil of elliptic curves.
\end{definition}

The significance of Manin transformations is that \cite[Lemma 4.2.1]{duistermaat10dis} $I_{B_1}\circ I_{B_2}$ is the unique translation on an elliptic curve that maps $B_1$ to $B_2$. Thus a dynamical system $x_n\mapsto x_{n+1}$ defined by a Manin transformation has solution $x_n = x_0 + n (B_2 - B_1)$, where the operations are those of the abelian group operation on the elliptic curve.



{\EBcolor We start by presenting a non-symmetric example.}
Let
\begin{equation*} 
 H(x,y) = \tfrac13 (x^3 + y^3) + \tfrac12 y^2.
\end{equation*}
We denote by $\Phi_h$ the Kahan map applied to this problem. When Kahan's method is applied to the corresponding problem, the modified Hamiltonian is
\begin{equation} \label{eq:nonsymm}
\widetilde{H}(x,y; h) = \frac{\tfrac13(x^3+y^3)+(\frac12-\tfrac1{12}h^2x)y^2}{1+h^2x(y+\tfrac12)} .
\end{equation}
\begin{proposition}
The Kahan map $\Phi_h$ is a Manin transformation. It can be written as a composition of involutions
\begin{equation*}
      \Phi_h = I_{B_1}\circ I_{B_\infty}
\end{equation*}
where the base point $B_1$ is finite and $B_{\infty}$ is a point at infinity with
\begin{equation*}
       I_{B_\infty}: (x,y) \mapsto (x_c,y_c)
\end{equation*}
where
\begin{equation*}\eqalign{
       x_c &= \frac{\beta x r^2-12y(y+1)r-h^2(y^2+6(2y+1)\widetilde{H})}{\beta r^2-2h^2(y+6\widetilde{H})r-12x}, \cr
       y_c&=y-r(x_c-x),
}\end{equation*}
and $r=r(h)$ is the unique positive real root of the cubic equation $4r^3+r^2h^2-4=0$, $\beta=h^2x-12y-6$ and $\widetilde{H}=\widetilde{H}(x,y; h)$.
\end{proposition}

\noindent \textbf{Remark.} We do not give the explicit expression for $I_{B_1}$ here; it can be obtained simply as $I_{B_1}=\Phi_h\circ I_{B_\infty}$.

\begin{proof}
 We obtain the base points as common zeros in the projective plane of  the numerator and denominator of $\widetilde{H}$. We can explicitly solve for 
$y$ in terms of $x$ from the denominator, substituting this into the numerator we find that the $x$-coordinate of the finite base points satisfy
$r=hx$ where $r$ is a root of
 \begin{equation} \label{altnumer}
( 4r^3+h^2r^2-4)(4r^3-h^2r^2+4)=0,
 \end{equation}
 and the $y$-coordinate satisfies $s=hy$ where
 \begin{equation*}
     s=-\frac{1+\tfrac12h r}{r}.
 \end{equation*}
Using Budan's theorem we find that for any $h$ there is precisely one real positive root of the factor  $4r^3+h^2r^2-4$; for ease of notation we denote this root by $r$ in what follows.
  The finite base point we use in the construction of the involution $I_{B_1}$ is $(r/h,s/h)$.
 
 Using homogeneous coordinates,  $y=Y/R$, $x=X/R$, we find for $R=0$ that 
 \begin{equation*}
    \tfrac13 (X^3+Y^3)-\tfrac1{12} h^2 XY^2 = 0.
 \end{equation*}
 Setting $z=Y/X$, we get
 \begin{equation*}
 4z^3- h^2 z^2+4 = 0
 \end{equation*}
 which therefore has precisely one negative real root, $z=-r$. This base point at infinity is used to construct $I_{B_\infty}$.
  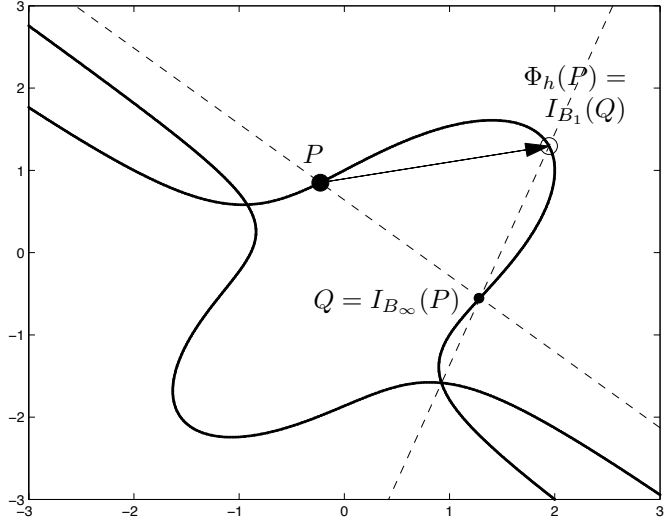
\begin{figure}[t]
  \begin{center}
 \begin{tikzpicture}
	\node[anchor=south west, inner sep=0] (image) at (0,0) {\includegraphics[width=0.6\textwidth]{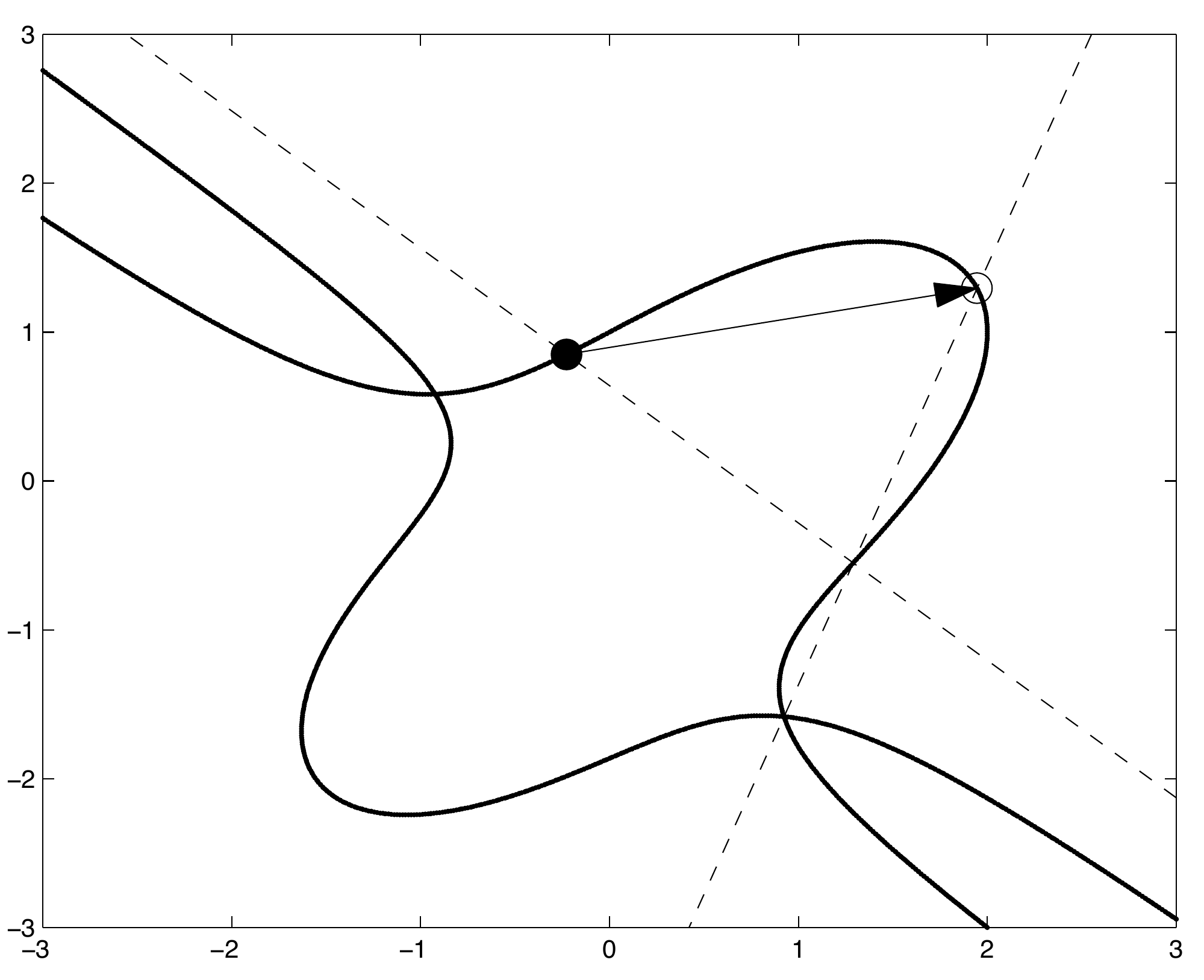}};
	\begin{scope}[x={(image.south east)},y={(image.north west)}]
	\node at (0.46, 0.69) {$P$};
	\fill [black] (0.707,0.425) circle (2pt);
	\node at (0.57,0.42) {$Q = I_{B_\infty}(P)$};
	\node at (0.85,0.8) {$\begin{array}{r}\Phi_h(P) = \\ I_{B_1}(Q)\end{array}$}; 
	\end{scope}
\end{tikzpicture}
	
\caption{Kahan mapping $\Phi_h$ associated with the Hamiltonian $H(x,y)=\frac{1}{3}(x^3+y^3) + \frac{1}{2}y^2$ at step size $h=1$. The mapping is represented by the arrow between the two points marked with circles. See text for the construction of the two involutions $I_{B_1}$ and $I_{B_\infty}$.} 
\label{fig:1}
\end{center}

\end{figure}

 The  construction is illustrated in Figure~\ref{fig:1} where two level curves of the non-symmetric modified Hamiltonian (\ref{eq:nonsymm}) with $h=1$
are plotted. These intersect each other in two finite base points. A third base point is at infinity. The dashed lines represent the two involutions $I_{B_1}$ and $I_{B_\infty}$,  and  $\Phi_h=I_{B_1}\circ I_{B_\infty}$. With reference to Definition~\ref{def1}, the source point $P$ of $I_{B_\infty}$ is the filled bullet, and the target $Q=I_{B_\infty}(P)$ is the intersection between the declining dashed line and the level curve. This is also the source point of $I_{B_1}$, whose target, $I_{B_1}(Q)$, is the intersection of the inclining dashed line with the level curve marked by an open bullet.
In this example $r\approx 0.923$.
 
We let $(x,y)$ and $h$ be fixed, but arbitrary. We now show that $\Phi_h=I_{B_1}\circ I_{B_\infty}$ and is therefore a Manin transformation.
 Compute $(x_c,y_c)$, the intersection between the two straight lines, by the formulas  
  \begin{eqnarray*}
 x' &= \frac{x+h(y^2+y)}{1+h^2(xy+\tfrac12 x)} & \mbox{(Kahan map\ $x$-coordinate),}\\[1mm]
 y' &= \frac{y-hx^2-\tfrac12h^2xy}{1+h^2(xy+\tfrac12x)} & \mbox{(Kahan map\  $y$-coordinate),} \\
 y_c &= y - r(x_c-x), \\[1mm]
 y_c &= y'+\frac{hy'-s}{hx'-r}(x_c-x'), \\[1mm]
 s&=-\frac{1+\tfrac12h r}{r}.
  \end{eqnarray*}
 Then $(x_c,y_c)$ depends on $r$ in addition to $x,y,h$. Now factor
$
      \widetilde{H}(x_c,y_c)-\widetilde{H}(x,y)
$
 and observe that this rational expression has a factor $   4r^3 + h^2 r^2 -4 $ in the numerator
 which is zero because of the definition of $r$. This proves that the intersection point $(x_c,y_c)$ lies on the same elliptic curve as $(x,y)$.
 \end{proof}
 
 \begin{figure}[t]
\begin{center}
\includegraphics[width=0.6\textwidth]{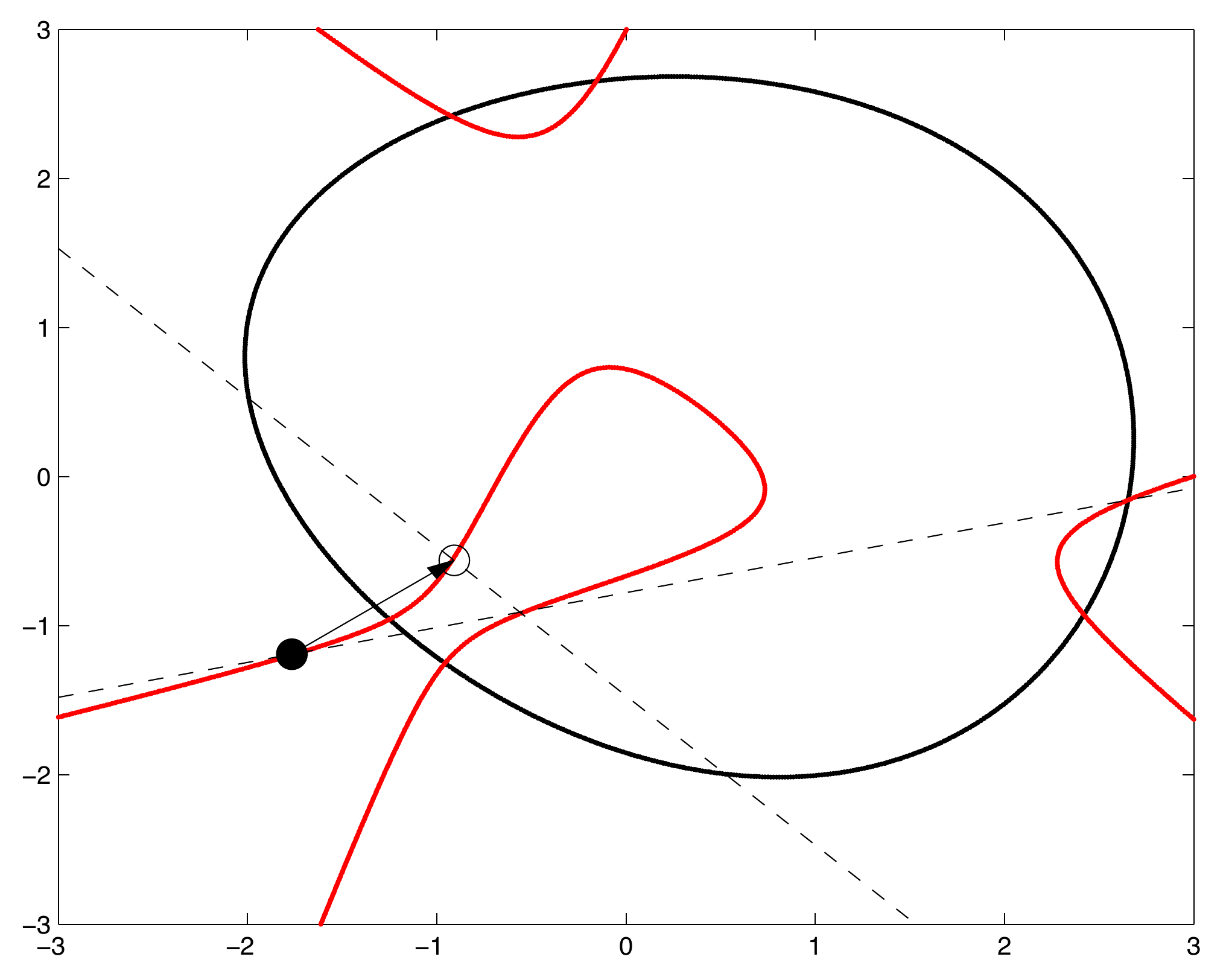}
 \caption{\label{fig:sym}A symmetric example of the construction of the Kahan mapping as a Manin transformation. The Hamiltonian is  (\ref{eq:2}) with $b=-0.0986$, $c=0.416$, $d=0.674$, $e=-0.428$, $g=0.134$, and the step size is$h=1$. The level sets $\widetilde H^{-1}(\infty)$ and $\widetilde H^{-1}(0.25)$ are shown in black and red, and intersect at base points.}
 \end{center}
 \end{figure}
 
We now proceed to show that for all {\it symmetric} Hamiltonians on $\mathbb{R}^2$, Kahan's method yields a Manin transformation. An example of the construction is shown in Figure \ref{fig:sym}.

\begin{proposition} \label{prop:symm}
The map $\Phi_h$ obtained by applying Kahan's method to the vector field \eqref{eq:1}
 with symmetric cubic Hamiltonian
 \begin{equation}
 \label{eq:2}
 H(x,y)=b(x^3+y^3)+c(x^2y+xy^2)+d(x^2+y^2)+exy+g(x+y)
 \end{equation}
 is given by
 \begin{equation}
 \label{eq:3}
 \Phi_h=I_{B_\infty}\circ I_{B_2},
 \end{equation}
 where
 \begin{equation}
 \label{eq:4}
 I_{B_\infty}\left(\matrix{x\cr y}
\right)=\left(\matrix{y\cr x}\right)
 \end{equation}
 and
 $I_{B_2}$ is the involution defined by the base point 
 \begin{equation}
 \label{eq:5}
 B_2=-\frac{1}{2h(3b-c)}\left(\matrix{ 2hd-he+2\cr 2hd-he-2} \right).
 \end{equation}
 \end{proposition}

\noindent \textbf{Remarks.} The involution $I_{B_\infty}$ corresponds to a base point at infinity, associated with the lines of slope $-1$. The fact that
$H(x,y)=H(y,x)$ implies that $\widetilde{H}(x,y;h)=\widetilde{H}(y,x;h)$ and thus all level curves of $\widetilde{H}$ must be symmetric with respect to the line $y=x$. This means that any two  points $(x,y), (y,x)$ lie on the same level curve, and are joined by a line of slope $-1$.
 
 For the special case $c=3b$, the Manin transformation becomes a composition of two base points at infinity. $I_{B_1}(x,y)$ is then obtained as the second intersection of the line through $(x,y)$ with slope
 \begin{equation*}
       \frac{2hd-he-2}{2hd-he+2}
 \end{equation*}
and the level curve of $\widetilde{H}$ on which $(x,y)$ lies.

 \begin{proof}
 The proof is obtained by direct computation.
   \end{proof}
 We have presented a non-symmetric example followed by the general case of symmetric Hamiltonians in 2 dimensions.
 Since Kahan's method is a Runge-Kutta method, it is covariant with respect to affine transformations \cite{mclachlan1998numerical}. If $H(u)$ is a symmetric Hamiltonian, we may use an affine change of variables $u=\varphi(v)=Av+b$ for a $2\times 2$ non-singular matrix $A$ and a 2-vector $b$.
 The corresponding Hamiltonian is $K(v)=H(\varphi(v))$. It is a well-known fact that if the same Runge-Kutta method is applied to each of these problems the corresponding approximations $\{u_n\}$, $\{v_n\}$ satisfy $u_n=\varphi(v_n),\ n>0$, if $u_0=\varphi(v_0)$. In particular, for Kahan's method, $\Phi_h$, we have established  that there are involutions $I_{B_\infty}$ and $I_{B_2}$ such that $u_{n+1}=I_{B_\infty}\circ I_{B_2}(u_n)$ so it now follows that
 \begin{equation*}
      v_{n+1} = \varphi(u_{n+1})=\varphi\circ I_{B_\infty}\circ\varphi^{-1}\circ\varphi\circ I_{B_2}\circ \varphi^{-1}(v_n)
      =\bar{I}_{B_\infty} \circ\bar{I}_{B_2}(v_n),
 \end{equation*}
 where the new maps $\bar{I}_{B_\infty}$ and $\bar{I}_{B_2}$ are again involutions of the same type as in Definition~\ref{def1}.
 Thus we can generalize the symmetric case of Proposition~\ref{prop:symm} to all Hamiltonians which are related to a symmetric one via an affine transformation. Unfortunately, this does not cover all the cases of non-symmetric Hamiltonians, and it can be proved that the non-symmetric example given in the beginning of this section cannot be transformed to a symmetric one in this way.

\subsection{Nambu systems}
\label{subs:nambu}

A (canonical) Nambu system in $\R^3$ is the ODE
\begin{equation}
\label{eq:nambu} \dot x = \nabla H_1(x)\times \nabla H_2(x)
\end{equation}
where $H_1$ and $H_2$ are functions on $\R^3$ \cite{nambu1973generalized,vaisman1999survey}. 
Nambu systems are completely integrable with first integrals $H_1$ and $H_2$ (the `Hamiltonians') and conserved measure $dx_1\wedge dx_2\wedge dx_3$. They can also be viewed as Poisson systems in either of the Poisson forms 
\begin{equation*} \dot x = \widehat{\nabla H_1(x)}\nabla H_2(x) = -\widehat{\nabla H_2(x)} \nabla H_1(x)\end{equation*}
where $\widehat{\ }\colon\R^3\to\R^{3\times 3}$ is given by
\begin{equation*}\hat v =\left(\matrix{0 & -v_3 & v_2 \cr v_3 & 0 & -v_1 \cr -v_2 & v_1 & 0 }\right).\end{equation*}

When $H_1$ and $H_2$ are quadratic, the Nambu system is quadratic and we can ask whether Kahan's method has conservative properties.

\def\adj{{\rm adj}}
\def\diag{{\rm diag}}
\def\tr{{\rm tr}}

For a square matrix $M$, let $\adj(M)$ be the adjugate of $M$, i.e., the transpose of the matrix of cofactors of $M$. If $M$ is invertible, $\adj(M) = \det(M) M^{-1}$.

\begin{proposition}
\label{prop:nambu}
Let $H_1=x^T A x$ and $H_2=x^T B x$ be homogeneous quadratics on $\R^3$. Let $C = A\, \adj(B) A$ and $H_3(x) = x^T C x$.
Then the Kahan method applied to the Nambu system
(\ref{eq:nambu}) has invariant measure
\begin{equation*} \frac{dx_1\wedge dx_2 \wedge dx_3}{(1 + 4 h^2 H_3(x))^2}\end{equation*}
and first integrals
\begin{equation*} \widetilde H_i(x) := \frac{H_i(x)}{1 + 4 h^2 H_3(x)},\quad i=1,2\end{equation*}
and hence is integrable.
Consequently the $h$-independent measure $dx_1\wedge dx_1\wedge dx_3/(H_1(x)H_2(x))$ and the $h$-independent function $H_1(x)/H_2(x)$ are also conserved. Note that the entries of $C$ are biquadratic functions of the entries of $A$ and $B$.
\end{proposition}

\begin{proof}
Using \eqref{eq:rosenbrock}, 
a direct substitution shows that $\widetilde H_i(x')=\widetilde H_i(x)$. The Jacobian derivative $D := \frac{\partial x'}{\partial x}$ of the Kahan mapping can be written
\begin{equation*} D  = \left(I -\frac{h}{2} f'(x)\right)^{-1}\left(I + \frac{h}{2} f'(x') \right).
\end{equation*}
Substituting $x'$ from (\ref{eq:rosenbrock}) into $D$ and using the explicit form of $f'(x)$ shows that
\begin{equation*} \det D= \frac{(1 + 4 h^2 H_3(x'))^2}{(1+ 4 h^2 H_3(x))^2}\end{equation*}
as required.
\end{proof}

Under linear changes of variables $y = L x$, Nambu's equations transform to
\begin{equation*} \dot y = L L^T (\nabla \overline{H_1}(y)\times \nabla\overline{H_2}(y))\end{equation*}
where $\overline{H_i}(y) = H_i(L^{-1}(y)) $.
Thus the form of Nambu's equations is invariant under orthogonal linear maps. Using such maps one or other of $A$, $B$ (but not, in general, both) can be diagonalized. Some special cases of note are
\begin{enumerate}
\item When $A = I$, $H_1 = \|x\|^2$, we have $C = \adj(B)$.
\item When $B=I$, $H_2 = \|x\|^2$, we have $C = A^2$.
\item When $A$ and $B$ are simultaneously diagonalizable, say $A=\diag(a_1,a_2,a_3)$ and $B=\diag(b_1,b_2,b_3)$, we have
\begin{equation*} C = \diag(a_1^2 b_2 b_3, a_2^2 b_1 b_3, a_3^2 b_1 b_2).\end{equation*}
This generalizes the 3-parameter example of \cite[Eq. (6.1)]{petrera12oio}, which includes the Euler free rigid body (see also \cite{petrera10,petrera12,petrera13}) and the Lagrange or Nahm system $\dot x_1 = x_2 x_3$, $\dot x_2 = x_3 x_1$, $\dot x_3 = x_1 x_2$ \cite{takhtajan1994foundation} to a linearly-invariant, 11-parameter family of quadratic vector fields.
\end{enumerate}
It is striking that Nambu's equations are essentially invariant under interchanging $A$ and $B$, but the formula $C=A\,\adj(B) A$ is not. However, it does define a striking product of quadratic forms.
When $A$ and $B$ lie in the symmetric space of symmetric positive definite matrices, it gives the image of $A$ under the geodesic symmetry at $B$. We do not develop this connection here. However, 
 the expressions for the conserved quantities can be written in a more invariant way as follows.

\begin{proposition}

Let $H_1=x^T A x$ and $H_2=x^T B x$ be homogeneous quadratics on $\R^3$. Let 
\begin{equation*}C = \frac{1}{2}\left(A\, \adj(B) A + B\, \adj(A) B\right) + \alpha(A\, \adj(B) A - B\, \adj(A) B) + \beta A + \gamma B\end{equation*}
and let $H_4(x) = x^T C x$.
Then for any constants $\alpha$, $\beta$, $\gamma$, 
the Kahan method applied to the Nambu system
(\ref{eq:nambu}) preserves the measure
\begin{equation*} \frac{dx_1\wedge dx_2 \wedge dx_3}{(1 + 4 h^2 H_4(x))^2}.\end{equation*}
\end{proposition}
\begin{proof}
First, if $\mu$ is an invariant measure, then so is $g(\widetilde H_1,\widetilde H_2)\mu$. Take $g(a,b)=(1+\alpha a + \beta b)^{-2}$ to get the invariant measure
\begin{equation}
\label{eq:mes2}
\eqalign{
 \frac{dx_1\wedge dx_2\wedge dx_3}  {(1 + 4 h^2 H_3)^2}& \left(1 + \frac{\alpha H_1}{1 + 4 h^2 H_3} + \frac{\beta H_2}{1 +4 h^2 H_3}\right)^{-2} \cr
&= \frac{dx_1\wedge dx_2\wedge dx_3}{(1 + 4 h^2 H_3)^2} 
\left(\frac{1 + 4 h^2 H_3 + \alpha H_1 + \beta H_2}{1 + 4 h^2 H_3}\right)^{-2} \cr
&= \frac{dx_1\wedge dx_2\wedge dx_3}{(1 + 4 h^2 H_3 + \alpha H_1 + \beta H_2)^2} \cr
}
\end{equation}
Next, we have that for any $3\times 3$ matrices $A$ and $B$, 
\begin{equation}
\label{eq:adjid}
 A\,\adj(B) A - B\,\adj(A) B = \tr(A\, \adj B)A - \tr(B\, \adj A)B,
 \end{equation}
which can be shown by multiplying out both sides.
Replacing $A\, \adj(B) A$ by $\frac{1}{2}A\, \adj(B) A+\frac{1}{2}A\, \adj(B) A$ and applying (\ref{eq:adjid}) and (\ref{eq:mes2}) gives the result.
\end{proof}

The Nambu systems in Prop. \ref{prop:nambu} are all 3-dimensional Lie--Poisson systems. There are 9 inequivalent families of real irreducible 3-dimensional Lie algebras \cite{patera76}. Five of them have homogeneous quadratic Casimirs and are covered by Prop. \ref{prop:nambu}: in the notation of \cite{patera76}, they are 
$A_{3,1}$ ($C=x_1^2$, Heisenberg Lie algebra)
$A_{3,4}$ ($C=x_1 x_2$, $\mathfrak{e}(1,1)$); 
$A_{3,6}$ ($C=x_1^2+x_2^2$, $\mathfrak{e}(2)$); 
$A_{3,8}$ ($C=x_2^2 + x_1 x_3$, $\mathfrak{su}(1,1)$, $\mathfrak{sl}(2)$; and
$A_{3,9}$ ($C=x_1^2+x_2^2+x_3^2$, $\mathfrak{su}(2)$, $\mathfrak{so}(3)$.
There are two other 3-dimensional Lie algebras for which Kahan's method applied to the associated Lie--Poisson system with (even nonhomogeneous) quadratic Hamiltonian is integrable. The first is $\mathfrak{a}(1)\times\R$, with Poisson tensor
\begin{equation*}\left(\matrix{0 & x_2 & 0 \cr -x_2 & 0 & 0 \cr 0 & 0 & 0 }\right).\end{equation*}
The planes $x_3=$ const. are invariant, and on each plane the system reduces to the 2-dimensional system 
\begin{equation}
\label{eq:2dlp}
\dot x = x_1 \left(\matrix{0 & 1 \cr -1 & 0 }\right)\nabla H
\end{equation}
 where $H$ is a nonhomogeneous quadratic. This is a 2-dimensional Suslov system; the integrability of the Kahan method applied to (\ref{eq:2dlp}) is established in Prop. \ref{prop:suslov2}.

The second is $A_{3,3}$, with Casimir $x_2/x_1$ and Poisson tensor
\begin{equation*}\left(\matrix{0 & 0 & x_1 \cr 0 & 0 & x_2 \cr -x_1 & -x_2 & 0 }\right).\end{equation*}
As the symplectic leaves are half-planes $x_2/x_1=$ const. and the points of the $x_3$-axis, these are
preserved by Kahan's method; on each half-plane, the equations of motion reduce to the 2-dimensional Suslov system (\ref{eq:2dlp}) whose Kahan map is integrable.


A related study is undertaken in \cite{hone09tds}, in which the Kahan mapping for several Nambu systems is integrated explicitly in terms of elliptic and transcendental functions.

\subsection{Suslov system}
\label{subs:suslov}
Consider the system
\begin{equation}
\label{eq:linSuslov}
\dot{x}=\ell (x)\, J\, \nabla H
\end{equation}
where $x\in \mathbb{R}^m$, $\ell(x) $ is a linear homogeneous polynomial, $J$ is a constant skew-symmetric $m\times m$ matrix and $H\colon\mathbb{R}^m\to\mathbb{R}$ is a quadratic homogeneous polynomial.

The Jacobian of the vector field $f(x)=\ell (x)\, J\, \nabla H$ is
\begin{equation*}f'(x)=J\,\nabla H \,\nabla\,  \ell (x)^T+\ell(x)J\, H''.\end{equation*} 

\begin{proposition}
Kahan's method applied to (\ref{eq:linSuslov}) 
preserves the measure
\begin{equation}
\label{eq:Suslovmeasure}
\frac{dx_1\wedge dx_2\dots \wedge dx_N}{\mathrm{det}(I-\frac{h}{2}\ell(x)JH'')\ell(x)}.
\end{equation}
\end{proposition}
\begin{proof}
Differentiating the Kahan map and taking determinants gives 
\begin{equation*}\mathrm{det}\frac{\partial x'}{\partial x}=\frac{\mathrm{det}(I+\frac{h}{2}f'(x'))}{\mathrm{det}(I-\frac{h}{2}f'(x))}=\frac{\mathrm{det}(I+\frac{h}{2}J\nabla H(x') \nabla\,  \ell (x')^T+\frac{h}{2}\ell(x')\, JH'')}{\mathrm{det}(I-\frac{h}{2}J\nabla H \nabla\,  \ell (x)^T-\frac{h}{2}\ell(x)\, JH'')}.\end{equation*}
Factorizing the term $\mathrm{det}(I+\frac{h}{2}\ell(x')\, JH'')$ in the numerator and the term $\mathrm{det}(I-\frac{h}{2}\ell(x')\, JH'')$ in the denominator (these two factors coincide because of the properties of $J$ and $H''$) gives
\begin{equation*}\mathrm{det}\frac{\partial x'}{\partial x}=\frac{\mathrm{det}(I+\frac{h}{2}\ell(x')\, JH'')\mathrm{det}(I+\frac{h}{2\ell(x')}v(x')u^T)}{\mathrm{det}(I-\frac{h}{2}\ell(x')\, JH'')\mathrm{det}(I-\frac{h}{2\ell(x)}w(x)u^T)}\end{equation*}
where 
\begin{equation*}\eqalign{v(x')&=\ell(x')(I+\frac{h}{2}\ell(x')\, JH'')^{-1}J\nabla H(x'),\cr
w(x)&=\ell(x)(I-\frac{h}{2}\ell(x)\, JH'')^{-1}J\nabla H(x),
}\end{equation*}
and $u=\nabla \ell.$
By the properties of rank-one perturbations of the identity, we have
\begin{equation*}\mathrm{det}\frac{\partial x'}{\partial x}=\frac{\mathrm{det}(I+\frac{h}{2}\ell(x')\, JH'')(1+\frac{h}{2\ell(x')}u^Tv(x'))}{\mathrm{det}(I-\frac{h}{2}\ell(x)\, JH'')(1-\frac{h}{2\ell(x)}u^Tw(x))}.\end{equation*}
The Kahan map in the form (\ref{eq:rosenbrock2}) 
gives 
\begin{equation*}u^Tv(x')=\frac{u^T(\frac{x'-x}{h})}{1-\frac{h}{2\ell(x')}u^T(\frac{x'-x}{h})}\end{equation*}
and in the form (\ref{eq:rosenbrock}) gives
\begin{equation*}u^Tw(x)=\frac{u^T(\frac{x'-x}{h})}{1+\frac{h}{2\ell(x)}u^T(\frac{x'-x}{h})},\end{equation*}
and so
\begin{equation*}\mathrm{det}\frac{\partial x'}{\partial x}=\frac{\mathrm{det}(I+\frac{h}{2}\ell(x')\, H'')}{\mathrm{det}(I-\frac{h}{2}\ell(x)\, H'')}\frac{\ell(x')(2\ell(x)+u^T(x'-x))}{\ell(x)(2\ell(x')-u^T(x'-x))}.\end{equation*}
Since $\ell$ is linear and homogeneous, $\ell(x)=\nabla\ell(x)^Tx=u^Tx$, and $u$ is a constant vector, so
\begin{equation*}2\ell(x')-u^T(x'-x))=2\ell(x)+u^T(x'-x))=u^T(x+x'),\end{equation*}
giving 
\begin{equation*}\mathrm{det}\frac{\partial x'}{\partial x}=\frac{\mathrm{det}(I+\frac{h}{2}\ell(x')\, JH'')}{\mathrm{det}(I-\frac{h}{2}\ell(x)\, JH'')}\frac{\ell(x')}{\ell(x)}\end{equation*}
as required.
\end{proof}

\begin{proposition}
Kahan's method applied to (\ref{eq:linSuslov}) has a conserved quantity given by 
\begin{equation}
\label{eq:Suslovintegral}
\widetilde{H}:= H(x)+\frac{1}{4}h\ell(x)\nabla H(x)^T\left(I-\frac{h}{2}J\ell(x) H''(x)\right)^{-1}J\,\nabla H(x).
\end{equation}
\end{proposition}
\begin{proof}
Let $\bar{x}=(x+x')/2$. Then
\begin{eqnarray*}
0&=h\left(-\tfrac{1}{2}\ell(x)\nabla H(x)+2\ell(\bar{x})\nabla H(\bar{x})-\tfrac{1}{2}\ell(x')\nabla H(x')\right)^TJ^T \\
& \qquad \left(-\tfrac{1}{2}\ell(x)\nabla H(x)+2\ell(\bar{x})\nabla H(\bar{x})-\tfrac{1}{2}\ell(x')\nabla H(x')\right)\\
 &=(x'-x)^T \left(-\tfrac{1}{2}\ell(x)\nabla H(x)+2\ell(\bar{x})\nabla H(\bar{x})-\tfrac{1}{2}\ell(x')\nabla H(x')\right))\\
 &=(x'-x)^T \left(\ell(x)\nabla H (x')+\ell(x')\nabla H(x)\right)),
\end{eqnarray*}
and so
\begin{equation}
\label{eq1}
\nabla H(x)^T(x'\ell(x)+x\ell(x'))=\nabla H(x')^T(x'\ell(x)+x\ell(x')).
\end{equation}
The sum of the vectors
\begin{equation*} v(x'):=(I+\frac{h}{2}\ell(x')JH'')^{-1}f(x'),\quad w(x'):=(I-\frac{h}{2}\ell(x')JH'')^{-1}f(x'),
\end{equation*}
is orthogonal to $\nabla H(x')$:  
\begin{eqnarray*}
\nabla H(x')^T(v(x')+ & w(x')) \\
& =2\ell(x')\nabla H(x')^T\left(I-\frac{1}{4}h^4\ell(x')^2(JH'')^2\right)^{-1}J\nabla H(x') \\
&=0.
\end{eqnarray*}
Further, we have
\begin{equation*}
\begin{array}{ll}
x'-x=h\left(I+\frac{h}{2\ell(x')}v(x')\nabla \ell^T\right)^{-1}v(x'), &  v(x')=\frac{2\, \ell(x')}{\ell(x'+x)}\frac{x'-x}{h}, \\
x''-x'=h\left(I-\frac{h}{2\ell(x')}w(x')\nabla \ell^T\right)^{-1}w(x'), & w(x')=\frac{2\ell(x')}{\ell(x''+x')}\frac{x''-x'}{h},
\end{array}
\end{equation*}
leading to
\begin{equation}
\label{eq3}
w(x')+v(x')=\frac{2\ell(x')\left(\ell(x'+x)(x''-x')-\ell(x''+x')(x'-x)\right)}{h\ell(x''+x')\ell(x'+x)}.
\end{equation}
Using now
\begin{eqnarray*}
w(x)&=\left(I-\frac{h}{2}\ell(x)JH''\right)^{-1}f(x) \\
& =\frac{2\ell(x)}{\ell(x'+x)}\frac{x'-x}{h}
\end{eqnarray*}
and (\ref{eq1}),
we can rewrite $\widetilde{H}(x)$ in the form
\begin{eqnarray}
\label{eq2}
\widetilde{H}(x)&=\frac{\nabla H(x)^T(x\ell(x')+x'\ell(x))}{\ell(x'+x)}\\
&=\frac{\nabla H(x')^T(x\ell(x')+x'\ell(x))}{\ell(x'+x)}.
\end{eqnarray}
Finally using (\ref{eq3})
and (\ref{eq2})
we obtain
\begin{eqnarray*}
\widetilde{H}(x)-\widetilde{H}(x')&=\nabla H(x')^T\left[\frac{x\ell(x')+x'\ell(x)}{\ell(x'+x)}-\frac{x'\ell(x'')+x''\ell(x')}{\ell(x''+x')} \right]\\
&=-h\,\nabla H(x')^T(w(x')+v(x')) \\
&=0.
\end{eqnarray*}

\end{proof}

\begin{proposition}
\label{prop:suslov2}
The Kahan method applied to the Suslov system (\ref{eq:linSuslov}) preserves the measure (\ref{eq:Suslovmeasure}) and integral (\ref{eq:Suslovintegral}) when $\ell(x)$ is a nonhomogeneous linear function, $H$ is a nonhomogeneous quadratic, and $\mathrm{rank}(J)=2$.
\end{proposition}
\begin{proof} 
$J$ may be taken in its normal form, reducing the situation to two dimensions. Now the problem only has finitely many parameters and the preservation of the measure and integral can be checked algebraically.
\end{proof}


\subsection{Generalized Ishii equations}

The Ishii system \cite{ishii90ppa} is
\begin{equation}
\label{Ishii}
\begin{array}{lcl}
\dot{x}& =&y \\
\dot{y}&=&z \\
\dot{z}&= &12xy.
\end{array}
\end{equation}
The following two functions 
\begin{eqnarray} 
H_1 &=z- 6 x^2 \label{eq:exinv1}\\ 
H_2  &=xz-\frac12 y^2 - 4 x^3 \label{eq:exinv2}
\end{eqnarray}
are first integrals of the system (\ref{Ishii}), see \cite{aron}.
Since the flow of this system is also volume preserving, the equations are completely integrable.
\begin{proposition}
\label{prop:ishii}
The Kahan method is volume preserving for this problem and it also has the following two invariants:
\begin{eqnarray}
\widetilde{H}_1 &= z - 6 x^2 + \frac{3}{2} h^2 y^2, \label{eq:numinv1} \\
\widetilde{H}_2 &= x z - \frac12 y^2 - 4 x^3 + h^2(x y^2 + \frac{1}{24}z^2). \label{eq:numinv2}
\end{eqnarray}
\end{proposition}
\begin{proof}
The proof is obtained with
a symbolic computing package.
\end{proof}
In the following we present a generalization of this system and some results regarding its discretization using Kahan's method; the proofs are all obtained with a symbolic computing package.

\begin{proposition}
\label{propoIshii1}
Consider the following divergence-free generalization of the Ishii  system 
\begin{equation}
\label{eqGIshii}
\begin{array}{lcl}
\dot{x} &= &-c_2 x + b_2 y + b_3 z \\
\dot{y} &= &c_1 x + c_2 y + c_3 z \\
\dot{z} &= &a_{11}x^2 + a_{12}xy + a_{22} y^2
\end{array}
\end{equation}
where $b_2, b_3, c_1, c_2, c_3, a_{11}, a_{12}, a_{22}$ are arbitrary parameters.
Kahan's method applied to this system preserves volume   if and only if the following two conditions are satisfied
\begin{eqnarray}
  b_2a_{11} + c_2a_{12} - c_1a_{22} &=0, \label{cond1}\\
  b_3^2\, a_{11} + b_3c_3\, a_{12} + c_3^2\, a_{22} &= 0.\label{cond2}
\end{eqnarray}
\end{proposition}
\begin{proof}
The proof is obtained with
a symbolic computing package.
\end{proof}
Alternatively, one may express the parameters $a_{ij}$  as
\begin{equation}\label{eq:aeq}
a_{11} = k A_2c_3,\quad a_{12} = -k(A_1c_3+A_2b_3),\quad a_{22}=k A_1 b_3,
\end{equation}
where $k$ is an arbitrary parameter.  Here 
 \begin{equation*}
 A_1=b_2c_3-b_3c_2,\quad A_2=c_2c_3+b_3c_1,\quad A_3=-(b_2c_1+c_2^2),
 \end{equation*}
 which are just the two-forms $\mathrm{d}x_i\wedge \mathrm{d}x_j$ applied to the two 3-vectors $(b_j)$ and $(c_j)$ where $b_1=-c_2$.
%
%
%
%
%
\begin{proposition}
Under the conditions (\ref{cond1}), (\ref{cond2}), the system (\ref{eqGIshii}) has the  invariants
\begin{eqnarray} 
H_1 &= z+\frac{k}2(c_3\,x-b_3\, y)^2,   \label{eq:genexinv1} \\
H_2&=\frac{k}3\, \left( c_{3}x-b_{3}y \right)^{3}+\frac{c_1}{2}\,{x}^{2}+c_{2}\,xy+c_{3}\,xz-\frac{b_2}{2}\,{y}^{2}-b_{3}\,yz.
 \label{eq:genexinv2}
\end{eqnarray}
The system is completely integrable.
\end{proposition}
\begin{proof}
The invariants are obtained with
a symbolic computing package. Since the flow of this system is also volume preserving, the system is completely integrable.
\end{proof}
\vskip0.3cm
\begin{proposition} Under the conditions \eqref{cond1}, \eqref{cond2},
Kahan's method applied to \eqref{eqGIshii} has the   invariants 
 \begin{eqnarray*}
 \widetilde{H}_1 &= z + \frac{k}2(c_3\,x-b_3\, y)^2 - \frac{k h^2}{8}\big( A_2\,x - A_1\,y \big)^2, \\
 \widetilde{H}_2 &= H_2 + \frac{h^2}{24}\bigg(A_3(-c_1 x^2-2c_2xy-2c_3xz+2b_3yz+b_2y^2)+(A_1c_3-A_2b_3)z^2\\
 & +k(-2c_3A_2^2 x^3+2A_2(b_3A_2+2c_3A_1)x^2 y -2A_1(c_3A_1+2b_3A_2)xy^2+2b_3A_1^2 y^3)\bigg).
 \end{eqnarray*}
Kahan's method applied to \eqref{eqGIshii} 
yields a completely integrable map.
 \end{proposition}
 \begin{proof}
The invariants are obtained with
a symbolic computing package. Since by Proposition~\ref{propoIshii1}, Kahan's method applied to \eqref{eqGIshii} is also volume preserving, we can conclude it yields a completely integrable map.
\end{proof}
 Notice that in the original Ishii equations, one has $A_1=1, A_2=0, A_3=0$ as well as $k=-12$, $b_2=c_3=1$, $b_3=c_1=c_2=0$.

\section{Non-autonomous problems}

\subsection{Riccati equations}

In a primary motivating example, Kahan \cite{kahan} considered the scalar equation $\dot x = x^2+t\,$ for which Kahan's method showed a remarkable ability to integrate through and converge past singularities. 
This is explained by the following result.

\begin{proposition}
Kahan's method applied to the scalar Riccati differential equation
\begin{equation*} \dot x = b(t) + 2 a(t) x - c(t) x^2,\end{equation*}
in which the coefficients are evaluated at any suitable point, is integrable.
\end{proposition}
\begin{proof}
Writing $a$ for the evaluation of the coefficients {\tempcolor at the point $t_{n+\frac{1}{2}} := (n+\frac{1}{2})h$} (e.g. $a=a(t_{n+1/2})$), etc., Kahan's method is
\begin{equation*} \frac{x'-x}{h} = b +  a (x+x') - c x x'\end{equation*}
whose solution is the M\"obius map
\begin{equation}
\label{eq:mob}
 x' = \frac{x + h(a x + b)}{1 + h (c x - a)}
 \end{equation}
which is  linearized by the change of variables $x=u/v$ to $u' = u + h(a u + b v)$, $v' = v + h ( c u - a v)$.
\end{proof}

Indeed, this is an example of what Schiff and Shnider \cite{schiff99ana}  call {\em M\"obius integrators} for Riccati equations. They give a geometric description of a (matrix) Riccati differential equation as a local coordinate version of an equation on a Grassmannian. In the present example, $x=u/v$ where $u$ and $v$ obey the linear nonautonomous system
\begin{equation}
\label{eq:gr}
\left(\matrix{ \dot u \cr \dot v } \right) = 
\left(\matrix{ a & b \cr c & d }\right) \left(\matrix{ u \cr v} \right)
\end{equation}
where $d=-a$.
Any linear integrator applied to this system yields a M\"obius integrator; Kahan's method for $x$ turns out to be equivalent to applying  Euler's method to (\ref{eq:gr}). This explains why Kahan was able to get such good results for $\dot x = x^2+t$ and to integrate successfully through singularities.

\begin{corollary}
Kahan's method applied to quadratic nonautonomous systems of the {\tempcolor special} form
\begin{equation*}\dot x_i = f_i(x_1,\dots,x_i,t),\quad i=1,\dots,n,\ x_i(t)\in\R\end{equation*}
is integrable.
\end{corollary}
\begin{proof}
Kahan's method is a Runge--Kutta method, so it preserves the linear foliations $x_1=\dots=x_k=$ const, and the reduced methods for the subsystems are also given by Kahan's method \cite{mclachlan03lgf}. The equation for $x_1$ is a scalar Riccati differential equation for which Kahan's method is integrable. Substituting this solution into $\dot x_2 = f_2(x_1,x_2,t)$ yields a scalar Riccati differential equation for $x_2$, and so on.
\end{proof}

Now let $x\in\R^{n\times m}$, $a(t)\in\R^{n\times n}$, $b(t)\in\R^{n\times m}$, $c(t)\in\R^{m\times n}$, $d(t)\in\R^{m\times m}$ and consider the matrix Riccati differential equation
\begin{equation}
\label{eq:mrde}
\dot x = a(t)x + b(t) - x c(t) x - x d(t)
\end{equation}
which is the projection of  (\ref{eq:gr}) under $x = u v^{-1}$. 
Kahan's method yields the Lyapunov equation
\begin{equation*} \frac{x'-x}{h} = \frac{1}{2}a (x+x') + b - \frac{1}{2}x c x' - \frac{1}{2} x' c x - \frac{1}{2}(x+x')d.\end{equation*}
for $x'$. 
Although numerical tests indicate that Kahan's method is also able to integrate through singularities for all smooth coefficient functions $a(t)$, $b(t)$, $c(t)$, and $d(t)$, we do not know if it is in general a M\"obius integrator. However, we do have the following.

\begin{proposition}
Kahan's method is integrable for the matrix Riccati differential equation (\ref{eq:mrde}) with $a=b=d=0$.
\end{proposition}
\begin{proof}
The transformation $y = x^{-1}$, which reduces $\dot x = - x c x$ to $\dot y = c$, commutes with the Kahan discretization. That is, the solution of the Kahan mapping in this case, which can be checked algebraically to be
\begin{equation*} x' = x(I+h c x)^{-1}\end{equation*}
written in terms of $y$ becomes $y' = y + h c$. Thus the solution of the Kahan mapping is
\begin{equation*} x_n = \left(x_0^{-1} +  h\sum_{i=0}^{n-1} c(t_{i+1/2})\right)^{-1}.\end{equation*}
\end{proof}
In other words, the Kahan method in this case is equivalent to applying Euler's method to (\ref{eq:gr}) and is thus a M\"obius integrator.

\subsection{First Painlev\'e equation}

\label{subs:painleve}

Consider the Painlev\'e I equation $\ddot{x}=6x^2+t,$ which can be written as a nonautonomous Hamiltonian system, and reads
\begin{eqnarray*}
\dot{x}&=&y,\\
\dot{y}&=&6x^2+t,
\end{eqnarray*}
with Hamiltonian
\begin{equation*}H(t,x,y)=\frac{1}{2}y^2-2x^3-xt.\end{equation*}
The original format of the Kahan method as proposed by Kahan (see \cite{petrera12oio}), cannot be applied to non-autonomous systems. Adding the equation $\dot{t}=1$ to the system, we get the autonomous system
\begin{equation}
\label{PIsystem}
\begin{array}{lcl} 
\dot{x}&=&y,\\
\dot{y}&=&6x^2+t,\\
\dot{t}&=&1,
\end{array}
\end{equation}
to which we apply Kahan's method.
\begin{proposition}
Kahan's method applied to \eqref{PIsystem} yields the map
\begin{equation}
\label{disPI}
x_{n-1}+x_{n+1}=\frac{(2x_n+h^2t_n)}{(1-3h^2x_n)},
\end{equation}
which is an integrable discretization of the first Painlev\'e equation.
\end{proposition}
\begin{proof}
Applying Kahan's method to the quadratic vector field \eqref{PIsystem}
we obtain
the map 
\begin{eqnarray}
\label{KahanPI}
\left(\begin{array}{c}
\frac{x_{n+1}-x_n}{h}\\
\frac{y_{n+1}-y_n}{h}
\end{array}\right)&=
\left(I-\frac{h}{2}JH'' (x_n,y_n)\right)^{-1}J\nabla\, H(t_{n+\frac{1}{2}},x_n,y_n),\\[0.2cm]
t_n &= n\, h,
\end{eqnarray}
where $t_{n+\frac{1}{2}}=t_n+\frac{h}{2}$. After some simple algebra this map can be shown to be equivalent to \eqref{disPI}.

Now letting
\begin{equation*}x_n=-u_n+\frac{1}{3h^2}\end{equation*}
and substituting in in \eqref{disPI} we get
\begin{equation*}u_{n+1}+u_{n-1}=-\frac{(\frac{t_n}{3}+\frac{2}{9h^4})}{u_n}+\frac{4}{3h^2}.\end{equation*}
This is a special version of
\begin{equation}
\label{DPainI}
u_{n+1}+u_{n-1}=\frac{(An+B)}{u_n}+C
\end{equation}
with a rational relation between $A$, $B$ and $C$. This is an integrable discretization of Painlev\'e I  which appeared in \cite{brezin} and \cite{grammaticos}; see also \cite{bonan}.

\end{proof}

There are indications that there may exist other integrable mappings related to Painlev\'e equations via the Kahan discretization. First, consider the Painlev\'e I equation in the form $\ddot x = 6 x^2 + A t$. Differentiation with respect to $t$ gives  $x^{(3)} = 12 x \dot x + A$. The Kahan method  applied to the first-order form
\begin{equation*}\eqalign{
\dot x &= y\cr
\dot y &= z\cr
\dot z &= 12 x y + A\cr
}\end{equation*}
preserves Euclidean volume and passes the entropy test. Letting
\begin{equation*} I = z - 6 x^2 + \frac{3}{2}h^2 y^2,\end{equation*}
the Kahan map obeys
\begin{equation*} I_{n+1} = I_n + h A\end{equation*}
giving a time-dependent integral. The two-dimensional nonautonomous map obtained by eliminating $z$ using the integral appears to be different from  (\ref{DPainI}). The case $A=0$ is the Kahan discretization of the Ishii system, which was shown to be integrable in Prop. \ref{prop:ishii}. Hence, the case $A\ne 0$ may
correspond to another integrable discretization of Painlev\'e I.

Second, differentiating again gives $x^{(4)} = 12 x \ddot x + 12 \dot x^2$. The Kahan method applied to the first-order form
\begin{equation*}\eqalign{
\dot x &= y \cr
\dot y &= u\cr
\dot u &= v\cr
\dot v &= 12 x u + 12 y^2\cr
}\end{equation*}
has a polynomial integral $v - 12 x y + 3 h^2 u y$, 
preserves the polynomial measure $m(x) = 1 - 3 h^2 x + \frac{3}{3}h^4 u$, and
 passes the entropy test, and hence may also correspond to another integrable discretization of Painlev\'e I.

Third, all 6 Painlev\'e equations were written as nonautonomous  planar Hamiltonian ODEs by Okamoto \cite{okamoto}. 
In each case, the Hamiltonian is a polynomial, and for Painlev\'e I, II, and IV it is cubic.
When written as 3-dimensional nonautonomous systems as in (\ref{PIsystem}), no new cases of Painlev\'e II or IV were found in which the Kahan discretization passed the entropy test. However, 
the Hamiltonian system
\begin{equation*}\eqalign{
\dot x &= 4 x y - (x^2+ 2 c t x + 2\theta_0)\cr
\dot y &= -2 y^2 + 2 x y + 2 c t y - \theta_\infty \cr
}\end{equation*}
is Okamoto's form of Painlev\'e IV when $c=1$. The mapping obtained by freezing $t$ at $t_{n+1/2}$, and then applying the Kahan discretization on $[t_n,t_{n+1}]$---which is different from that obtained from the autonomizing version used above---was found to pass the entropy test for integrability in the cases (i) 
 $c=2$, $\theta_0$ arbitrary, $\theta_\infty=0$
and (ii) $c=-2$, $\theta_0=0$, $\theta_\infty$ arbitrary.
Thus, the Kahan method may generate integrable maps in these cases.


%
%

\section{Conclusion}

In this paper we have extended and generalized the impressive list of
cases presented in \cite{petrera12oio}, for which Kahan's method preserves integrability.
In particular, we have expanded the list beyond maps that are integrable
in terms of elliptic or hyperelliptic functions, to include our final
example, whose discretization is an integrable Painlev\'e equation. It still remains to actually integrate the Kahan map in many of these integrable cases as done here in section~\ref{sec:2.1} for the cubic planar Hamiltonian case.
In a
future paper, we hope to generalize our work on quadratic differential
equations, in the current paper and in \cite{CMOQgeometricKahan}, to the case of cubic and
higher-order polynomial differential equations.
Echoing the sentiment expressed in \cite{petrera12oio}, it is our hope that the present
work may help the Kahan--Hirota--Kimura discretization to attract the
attention of experts in integrable systems and in algebraic geometry, as
well as in geometric numerical integration.

\section*{Acknowledgements}
We are grateful to Chris Ormerod for very useful
correspondence regarding Subsection 3.1 and the related literature. This research was supported by a Marie Curie International Research Staff Exchange Scheme Fellowship within the 7th European Community Framework Programme, by the Marsden
Fund of the Royal Society of New Zealand, and by the Australian Research Council. 
Part of the work of G.R.W. Quispel was done while serving as Lars Onsager Professor at Norwegian University of Science and Technology, Trondheim, Norway.

\end{document}